\documentclass[10pt]{amsart}

\usepackage{amsmath, amscd}
\usepackage{amsfonts}
\usepackage{amssymb}
\usepackage{graphicx}
\usepackage{amsthm}
\newtheorem{definition}{Definition}[section]
\newtheorem{theorem}{Theorem}[section]
\newtheorem{lemma}{Lemma}[section]
\newtheorem{corollary}[theorem]{Corollary}
\newtheorem{proposition}{Proposition}

\newcommand{\st}{^*}

\newcommand{\mk}{\mathfrak}
\newcommand{\mc}{\mathcal}
\newcommand{\sh}{^{\sharp}}
\newcommand{\R}{\mathbb{R}}

\begin{document}
%\begin{titlepage}
%\begin{flushright}

%\baselineskip=12pt
%DSF--26--2008\\
%ICCUB-08-134
%\hfill{ }\\
%October 2008
%\end{flushright}

\title[]{Uniqueness of the Momentum map}

\author{Chiara Esposito} 
\address{Chiara Esposito, Department of Mathematical Sciences,
Universitetsparken 5,
DK-2100 Copenhagen {\O},
Denmark}
\email{esposito.chiar@gmail.com}

\author{Ryszard Nest}
\address{Ryszard Nest, Department of Mathematical Sciences,
Universitetsparken 5,
DK-2100 Copenhagen {\O},
Denmark}
\email{rnest@math.ku.dk}

\thanks{The first author was supported by the Danish National Research Foundation (DNRF) through the Centre for Symmetry and Deformation}

\maketitle

\begin{abstract}
We give a detailed discussion of existence and uniqueness of Lu's momentum map. We introduce the infinitesimal momentum map, and analyze its integrability to the usual momentum map, its existence and its  deformations.
\end{abstract}

%\end{titlepage}
\vskip 1 cm

\tableofcontents
\section{Introduction}

The classical momentum map for an action of a Lie group on a Poisson manifold provides a mathematical
formalization of the notion of conserved
quantity associated to symmetries of a dynamical system. The standard definition of momentum map  only requires a canonical Lie algebra action and its existence is guaranteed whenever the infinitesimal generators of the Lie algebra action are
Hamiltonian vector fields (modulo vanishing of a certain Lie algebra cohomology class). In this paper we  focus on a generalization of the momentum map provided by Lu \cite{Lu3}, \cite{Lu1}.

The detailed construction of this generalized momentum map and its basic properties are recalled in the 
following section. The basic structure is as follows. Given  Poisson Lie group $(G,\pi_G )$ one introduces 
the dual Poisson Lie group $(G^*,\pi_{G^*})$ and, under fairly general conditions, $G^*$ carries a Poisson 
action of $G$ (and vice versa). The Lie algebra ${\mathfrak g}$ of $G$ is naturally identified with the 
space of $G^*$-(left-)invariant one-forms on $G^*$:
$$
	\alpha : {\mathfrak g}\to \Omega^1 (G^* )^{G^*},
$$
Given a Poisson manifold $(M,\pi)$ with a Poisson action of $G$, a momentum map is a smooth, Poisson map
$$
	\boldsymbol{\mu} : M \to G^*
$$
satisfying
$$
	X_\xi = \pi^\sharp (\boldsymbol{\mu}^*(\theta_{\xi}))
$$
where $X$ is the map ${\mathfrak g} \to Vect (M)$ induced by the action of $G$ on $M$. A canonical example of a momentum map is the identity
map $G^*\to G^*$, in which case $\alpha$ coincides with the structure one-form in $\theta \in \Omega^1 (G^*, {\mathfrak g}^* )^{G^*}$ of the Lie group  $G^*$.

The Poisson structure on $G$ gives its Lie algebra a structure of a Lie bialgebra 
$({\mathfrak g},[\cdot,\cdot], \delta )$ and hence a structure of Gerstenhaber algebra on $\wedge^{\bullet}{\mathfrak g}$. 
On the other hand the Poisson bracket on $M$ gives $\Omega^{\bullet}(M)[1]$ a structure of Lie algebra with
bracket $[\cdot,\cdot]_{\pi}$ which induces a structure of Gerstenhaber algebra on $\Omega^{\bullet}(M)$. The map $\alpha$ from above lifts to a morphism of Gerstenhaber algebras
$$
 	\alpha : (\wedge^{\bullet}\mathfrak{g} ,\delta, [\;,\;])\longrightarrow (\Omega^{\bullet} (M), d_{dR},[\;,\;]_\pi ).
$$
which we will call an {\em infinitesimal momentum map} (cf. the subsection \ref{sub: structure} and the proposition \ref{prop:gerst}). The existence of the infinitesimal momentum map has been discussed in \cite{Gi1}. The main subject of this paper is the study of the properties of this infinitesimal momentum map and its relation to the usual momentum map. In particular, we show under which conditions it integrates to the usual momentum map.

The fact that $\alpha$ is  map of Gerstenhaber algebras reduces to two equations
$$
	\alpha_{[\xi,\eta]}=[\alpha_{\xi},\alpha_\eta ]_{\pi} \quad\mbox{ and }\quad
	d\alpha_{\xi} +\frac{1}{2}\alpha\wedge\alpha\circ\delta(\xi)=0
$$
The second is a Maurer-Cartan type equation, in fact, in the case when $M=G^*$ it is precisely the Maurer Cartan equation for the Lie group $G^*$. In the case when $\Omega^\bullet$ is formal, the second equation admits explicit solution modulo gauge equivalence (cf. Theorem  \ref{thm:formal}).

\medskip

\noindent {\bf Theorem} {\em Suppose that $M$ is is a K\"{a}hler manifold. The set of gauge equivalence classes of $\alpha\in \Omega^1 (M,\mk{g}^*)$ satisfying the equation
\begin{equation}
	d\alpha_{\xi} +\frac{1}{2}\alpha\wedge\alpha\circ\delta(\xi)=0	
\end{equation}
is in bijective correspondence with the set of the cohomology classes $c\in H^1(M,\mk{g}^*)$ satisfying
\begin{equation}
	[c,c]=0.
\end{equation}}

\medskip

The following describes conditions under which an infinitesimal momentum map integrates to the usual momentum map (cf. Theorem \ref{thm: rec} for the details). 

\medskip

\noindent {\bf Theorem }{\em
Let $(M,\pi )$ be a Poisson manifold and $\alpha : {\mathfrak g}\to \Omega^1 (M)$ an infinitesimal momentum map. Suppose that $M$ and $G$ are simply connected and $G$ is compact. Then 
 ${\mathcal D} = \{ \alpha_{\xi}-\theta_{\xi},\ \xi\in {\mathfrak g}\}$ generates an involutive distribution  on $M \times G^*$
  and a leaf $\boldsymbol{\mu}_{\mathcal{F}}$ of $\mathcal D$ is a graph of a momentum map if
\begin{equation}\label{eq:vp}
	\pi (\alpha_{\xi} ,\alpha_{\eta})-\pi_{G^*}(\theta_{\xi},\theta_{\eta})|_{\mathcal F} = 0 , \quad \xi ,\eta \in {\mathfrak g}
\end{equation}
}

\medskip

In  the section \ref{sec: rec} we study concrete cases of  this globalization question and
prove the existence and uniqueness/nonuniqueness of a momentum map associated to a given infinitesimal momentum map for the particular case when the dual Poisson Lie group is abelian, respectively the Heisenberg group. For the second case the result is as follows (cf. Theorem \ref{thm:heisenberg}).

\medskip

\noindent {\bf Theorem} {\em
Let $G$ be a Poisson Lie group acting on a Poisson manifold $M$ with an infinitesimal momentum map $\alpha$ and such that $G^*$ is the Heisenberg group.
Let $\xi,\eta,\zeta$ denote the basis of $\mathfrak{g}$ dual to the standard basis $x,y,z$ of $\mathfrak{g}^*$, with $z$ central and $[x,y]=z$. Then
\begin{equation}
\pi(\alpha_{\xi}, \alpha_{\eta})=c
\end{equation}
where $c$ is a constant on $M$. The form $\alpha$ lifts to a momentum map $\boldsymbol{\mu}:M\rightarrow G^*$ if and only if $c=0$. When $c=0$ the set of momentum maps with given $\alpha$ is one dimensional with free
transitive action of $\R$.
}

\medskip

Finally, in the last section, we study the question of  infinitesimal deformations of a given momentum map. The main result is Theorem \ref{thm: idef}, which describes explicitly the space tangent to the space of momentum maps at a given point. The main result can be formulated as a statement that the space of momentum maps has a structure of flat manifold (in an appropriate $C^\infty$ topology).

\medskip

\noindent {\bf Theorem} {\em 
Infinitesimal deformations of a momentum map are given by smooth maps $H:M\rightarrow {\mathfrak g}^*$ satisfying the equations

\noindent For all $\xi, \eta \in \mathfrak g$,
\begin{align}
	\label{eq: inf1}   X_{\xi} H (\eta)-X_{\eta} H (\xi) & = H( [\xi,\eta])\\
	\label{eq: inf2}   \{ H(\xi),\ \cdot \} & = -X_{ad^*_H\xi}.
\end{align}
}

This theorem has  the following corollary (cf. Corollary \ref{cor:unique}).
 
 \medskip
 
 \noindent{\bf Theorem} {\em Suppose that $G$ is a compact and semisimple Poisson Lie group
with Poisson action  on a Poisson manifold $M$ and with a momentum map $\boldsymbol \mu$. Any smooth deformation of $\boldsymbol \mu$ is  given by integrating a Hamiltonian flow on $M$ commuting with the action of $G$.}

\section{Preliminaries: Poisson actions and Momentum maps}\label{sec: pam}

In this section we give a brief summary of the notions of Poisson action and momentum map in the Poisson context. We discuss the dressing transformations as an example of Poisson actions that will allow us to introduce the concept of Hamiltonian action.

Recall that a Poisson Lie group $(G,\pi_G)$ is a Lie group equipped with a multiplicative Poisson structure $\pi_G$. From the Drinfeld theorem \cite{Drinfeld}, given a Poisson Lie group $(G,\pi_G)$, the linearization of $\pi_G$ at $e$ defines a Lie algebra structure on $\mathfrak{g}^*$ such that $(\mathfrak{g},\mathfrak{g}^*)$ form a Lie bialgebra over $\mathfrak{g}$. For this reason, in the following we always assume that $G$ is connected and simply connected.
\begin{definition}
 The action of $(G,\pi_G)$ on $(M,\pi)$ is called \textbf{Poisson action} if the map $\Phi:G\times M\rightarrow M$ is Poisson, where $G\times M$ is a Poisson  product with structure $\pi_G\oplus\pi$
\end{definition}

Given an action $\Phi:G\times M\to M$, we denote by $X_{\xi}$ the Lie algebra anti-homomorphism from $\mathfrak{g}$ to $M$ which defines the infinitesimal generator of this action.

\begin{proposition}
Assume that $(G,\pi_G)$ is a connected Poisson Lie group. Then the action $\Phi: G\times M\to M$ is a Poisson action if and only if
\begin{equation}\label{eq: pa}
	\mathcal{L}_{X_{\xi}}(\pi) = (\Phi\wedge\Phi)(\delta(\xi)),
\end{equation}
for any $\xi\in\mathfrak{g}$, where $\delta=d_e\pi_G:\mathfrak{g}\to \mathfrak{g}\wedge \mathfrak{g}$
is the derivative of $\pi_G$ at $e$.
\end{proposition}
The proof of this Proposition can be found in \cite{LW}. Motivated by this fact, we introduce the following definition.

\begin{definition}
A Lie algebra action $\xi\mapsto X_{\xi}$ is called an \textbf{infinitesimal Poisson action} of the Lie bialgebra $(\mathfrak{g},\delta)$ on $(M,\pi)$ if it satisfies eq. (\ref{eq: pa}).
\end{definition}

In this formalism the definition of momentum map reads (Lu, \cite{Lu3}, \cite{Lu1}):

\begin{definition}\label{def: mm}
A \textbf{momentum map} for the Poisson action $\Phi:G\times M\to M$ is a map $\boldsymbol{\mu}: M\rightarrow G^*$ such that
\begin{equation}\label{eq: mmp}
	X_{\xi} = \pi^{\sharp}(\boldsymbol{\mu}^*(\theta_{\xi}))
\end{equation}
where $\theta_{\xi}$ is the left invariant 1-form on $G^*$ defined by the element $\xi\in\mathfrak{g}=(T_e G^*)^*$ and $\boldsymbol{\mu}^*$ is the cotangent lift $T^* G^*\rightarrow T^*M$.
\end{definition}

\subsection{Dressing Transformations}\label{sec: dressing}

One of the most important example of Poisson action is the dressing action of $G$ on $G^*$. Consider a Poisson Lie group $(G,\pi_G)$, its dual $(G^*,\pi_{G^*})$ and its double $\mathcal{D}$, with Lie algebras
$\mathfrak{g}$, $\mathfrak{g}^*$ and $\mathfrak{d}$, respectively.

Let $L_{\xi}$ the vector field on $G^*$ defined by
\begin{equation}\label{eq: idr}
	L_{\xi} = \pi_{G\st}\sh(\theta_{\xi})
\end{equation}
for each $\xi\in\mathfrak{g}$. Here $\theta_{\xi}$ is the left invariant 1-form on $G^*$ defined by $\xi\in\mathfrak{g}=(T_eG^*)^*$. The map $\xi\mapsto L_{\xi}$ is a Lie algebra anti-homomorphism. Using the Maurer-Cartan equation for $G\st$:
\begin{equation}\label{eq: mct}
    d\theta_{\xi} + \frac{1}{2}\theta\wedge\theta\circ\delta(\xi) = 0
\end{equation}
The action $\xi\mapsto L_{\xi}$ is an infinitesimal Poisson action of the Lie bialgebra $\mathfrak{g}$ on the Poisson Lie group $G^*$, called left infinitesimal dressing action (see, for example, \cite{Lu3}).
Similarly, the right infinitesimal dressing action of $\mathfrak{g}$ on $G^*$ is defined by $R_{\xi}=-\pi_{G\st}\sh(\theta_{\xi})$
where $\theta_{\xi}$ is the right invariant 1-form on $G^*$.

Let $L_{\xi}$ (resp. $R_{\xi}$) a left (resp. right) dressing vector field on $G^*$. If all the dressing vector fields are complete, we can integrate the $\mathfrak{g}$-action into a Poisson $G$-action on $G^*$ called the
\textbf{dressing action} and we say that the dressing actions consist of dressing transformations. The orbits of the dressing actions are precisely the symplectic leaves in $G$ (see \cite{We1}, \cite{Lu3}).

The momentum map for the dressing action of $G$ on $G^*$ is the opposite of the identity map from $G^*$ to itself.

\begin{definition}
A multiplicative Poisson tensor $\pi$ on $G$ is complete if each left (equiv. right) dressing vector field is complete on $G$.
\end{definition}
It has been proved in \cite{Lu3} that a Poisson Lie group is complete if and only if its dual Poisson Lie group is complete.
Assume that $G$ is a complete Poisson Lie group. We denote respectively the left (resp. right) dressing action of $G$ on its dual $G^*$ by $g\mapsto L_g$ (resp. $g\mapsto R_g$).

\begin{definition}
A momentum map $\boldsymbol{\mu}:M\rightarrow G^*$ for a left (resp. right) Poisson action $\Phi$ is called \textbf{G-equivariant} if it is such with respect to the left dressing action of $G$ on $G^*$, that is,
$\boldsymbol{\mu}\circ \Phi_g = L_g\circ \boldsymbol{\mu}$ (resp. $\boldsymbol{\mu}\circ \Phi_g = R_g\circ \boldsymbol{\mu}$)
\end{definition}
A momentum map is $G$-equivariant if and only if it is a Poisson map, i.e. $\boldsymbol{\mu}_*\pi=\pi_{G^*}$. Given this generalization of the concept of equivariance introduced for Lie group actions, it is natural to call \textbf{Hamiltonian action} a Poisson action induced by an equivariant momentum map.

\section{The infinitesimal momentum map}\label{sec: uni}

In this section we study the conditions for the existence and the uniqueness of the momentum map. 
In particular, we give a new definition of the momentum map, called infinitesimal, in terms of one-forms 
and we study the conditions under which the infinitesimal momentum map determines a momentum map in the 
usual sense. We describe the theory of reconstruction of the momentum map from the infinitesimal one in two 
explicit cases. Finally, we provide the conditions which ensure the uniqueness of the momentum map.

\subsection{The structure of a momentum map}
\label{sub: structure}

Recall that, for the Poisson Lie group $G\st$ we identify $\mk{g}$ with the space of left invariant 1-forms 
on $G\st$;  this space is closed under the bracket defined by $\pi_{G\st}$ and the induced bracket on 
$\mk{g}$, by the above identification, coincides with the original Lie bracket on $\mk{g}$ (see \cite{We}).

\begin{proposition}
Let $\theta_{\xi},\theta_{\eta}$ be two left invariant 1-forms on $G\st$, such that $\theta_\xi(e)=\xi$, $\theta_{\eta}(e)=\eta$ then
\begin{equation}\label{eq: theta}
	\theta_{[\xi,\eta]} = [\theta_{\xi},\theta_{\eta}]_{\pi_{G\st}}
\end{equation}
and
\begin{equation}\label{eq: theta2}
 	\mc{L}_X\pi_{G\st}(\theta_{\xi},\theta_{\eta}) = x([\xi,\eta]) + \pi_{G^*}(\theta_{ad^*_x \xi},\theta_{\eta}) + \pi_{G^*}(\theta_{\xi},\theta_{ad^*_x\eta})
\end{equation}

\end{proposition}

\begin{proof}
 Let us consider and element $x\in\mk{g}\st$ and the correspondent left invariant vector field $X\in TG\st$.
Recall that given a Poisson manifold, the Poisson structure always induces a Lie bracket on the space of one-form on the manifold (see \cite{V}) by
\begin{equation}\label{eq: bff}
 	[\alpha,\beta]_{\pi} = \mc{L}_{\pi\sh(\alpha)}\beta-\mc{L}_{\pi\sh(\beta)}\alpha-d(\pi(\alpha,\beta)).
\end{equation}
Using this explicit formula for $[\theta_{\xi},\theta_{\eta}]_{\pi_{G\st}}$
 we can see that
\begin{equation}
	\iota_X [\theta_{\xi},\theta_{\eta}]_{\pi_{G\st}} = (\mc{L}_X\pi_{G\st})(\theta_{\xi},\theta_{\eta}).
\end{equation}
This proves that  $[\theta_{\xi},\theta_{\eta}]_{\pi_{G\st}}$ is a left invariant 1-form. In particular, since $\mc{L}_X\pi_{G\st}(e) = {}^{t}\delta(x)$, eq. (\ref{eq: theta}) is proved\footnote{This relation has already been claimed in \cite{YK}}.
Moreover, we have
\begin{equation}
	\begin{split}
		\mc{L}_X\pi_{G\st}(\theta_{\xi},\theta_{\eta}) & =(\mc{L}_X\pi_{G\st})(\theta_{\xi},\theta_{\eta}) + \pi_{G\st}(\mc{L}_X\theta_{\xi}, \theta_{\eta}) + \pi_{G\st}(\theta_{\xi},\theta_{\eta})\\
		& = {}^{t}\delta(x)(\xi,\eta) + \pi_{G^*}(\theta_{ad^*_x \xi},\theta_{\eta}) + \pi_{G^*}(\theta_{\xi},\theta_{ad^*_x\eta}),
	\end{split}
\end{equation}
since $\mc{L}_X\theta_{\xi} = \theta_{ad^*_x \xi}$. From ${}^{t}\delta(x)(\xi,\eta) = x([\xi,\eta])$, eq. (\ref{eq: theta2}) follows.
\end{proof}

As a direct consequence, recalling that the pullback and the differential commute and using the equivariance of the momentum map, we have the following proposition:
\begin{proposition}\label{prop:gerst}
Given a Poisson action $\Phi:G\times M\to M$ with equivariant momentum map $\boldsymbol{\mu}:M\to G\st$, the forms $\alpha_{\xi}=\boldsymbol{\mu}^*(\theta_{\xi})$ satisfy the following identities:
\begin{align}
%L_{\xi} &=\pi^{\sharp}(\alpha_{\xi}) \\
	\label{eq: ah}  \alpha_{[\xi,\eta]} & = [\alpha_{\xi},\alpha_\eta ]_{\pi}\\
	\label{eq: mca}  d\alpha_{\xi}  & + \frac{1}{2}\alpha\wedge\alpha\circ\delta(\xi)=0
\end{align}
\end{proposition}

This motivates the following Definition.
\begin{definition}\label{def: inf}
Let $M$ be a Poisson manifold and $G$ a Poisson Lie group. An \textbf{infinitesimal momentum map} is a morphism of Gerstenhaber algebras
\begin{equation}\label{eq: imm}
 	\alpha : (\wedge^{\bullet}\mathfrak{g} ,\delta, [\;,\;])\longrightarrow (\Omega^{\bullet} (M), d_{DR},[\;,\;]_\pi ).
\end{equation}
\end{definition}

The following theorem describes the conditions in which an infinitesimal momentum map determines a momentum map in the usual sense.

\begin{theorem}\label{thm: rec}
Let $(M,\pi )$ be a Poisson manifold and $\alpha:{\mathfrak g}\to \Omega^1 (M)$ a linear map which satisfies the
 conditions (\ref{eq: ah})-(\ref{eq: mca}). Then:
\begin{itemize}
  \item[(i)] The set $\{ \alpha_{\xi}-\theta_{\xi},\, \xi\in {\mathfrak g}\}$ generate an involutive distribution $\mathcal D$ on $M \times G^*$.
   \item[(ii)] If $M$ is connected and simply connected, the leaves $\mathcal{F}$ of  $\mathcal D$ coincide with the
    graphs of the maps $\boldsymbol{\mu}_{\mathcal{F}}: M\to G^*$ satisfying
$\alpha = \boldsymbol{\mu}^*_{\mathcal{F}} (\theta)$ and $G^*$ acts freely and transitively on the space of leaves by left multiplication on the second factor.
    \item[(iii)] The vector fields $\pi^{\sharp} (\alpha_{\xi})$ define a homomorphism from $\mathfrak{g}$ to $TM$.
If  they integrate to the action $\Phi : G\times M\to M$ (e.g. when $M$ is compact and $G$ simply connected), then $\Phi$ is a Poisson action
and $\boldsymbol{\mu}_{\mathcal{F}}$ is its momentum map if and only if the functions
\begin{equation}\label{eq:vp}
	\varphi (\xi,\eta) = \pi (\alpha_{\xi} ,\alpha_{\eta}) - \pi_{G^*}(\theta_{\xi},\theta_{\eta})
\end{equation}
satisfy
\begin{equation}\label{eq:vanishing}
	\varphi (\xi,\eta)\vert_{\mathcal{F}}=0
\end{equation}
for all $\xi,\eta\in \mathfrak{g}$.
\end{itemize}

\end{theorem}

\begin{proof}
\begin{itemize}
	\item[(i)] Using the eqs. (\ref{eq: mct}) and (\ref{eq: mca}), the $\mathfrak{g}$-valued form $\alpha -\theta$ on $M\times G^*$ satisfies
$
 	d(\alpha -\theta ) = (\alpha -\theta )\wedge (\alpha -\theta );
$
as a consequence, from the Frobenius theorem, it defines a distribution on $M\times G^*$. Let $\mathcal{F}$ be any of its leaves and let $p_i$, $i=1,2$ denote the projection onto the first (resp. second) factor in $M\times G^*$. Since the linear span of
$\theta_{\xi}$, $\xi\in \mathfrak{g}$ at any point $u\in G^*$ coincides with $T_u^*{G^*}$, the restriction of the projection $p_1 : M\times G \to M$ to $\mathcal{F}$ is an immersion.
Finally, since $dim(M) = dim(\mathcal{F})$, $p_1$ is a covering map.
	\item[(ii)] Under the hypothesis that $M$ is simply connected, $p_1$ is a diffeomorphism and
$$
 	\boldsymbol{\mu}_{\mc{F}} = p_2 \circ p_1^{-1}
$$
is a smooth map whose graph coincides with $\mathcal{F}$. It is immediate, that $\alpha =\boldsymbol{\mu}_{\mathcal{F}}^* (\theta)$. Moreover, since $\theta$'s are left invariant it follows 
immediately that the action of $G^*$ on the space of leaves by left multiplication of the second factor is 
free and transitive.
	\item[(iii)] Suppose that the condition (\ref{eq:vanishing}) is satisfied. Then
$$
	 \pi (\alpha_{\xi},\alpha_{\eta}) = \boldsymbol{\mu}^*_{\mathcal{F}} (\pi_{G^*}(\theta_{\xi},\theta_{\eta}))
$$
and $Ker {\boldsymbol{\mu}_{\mathcal{F}}}_*$ coincides with the set of zero's of $\alpha_{\xi},\ {\xi}\in \mathfrak{g}$. Hence, $\boldsymbol{\mu}_{\mathcal{F}}$ is a Poisson map and, in particular
$$
	{\boldsymbol{\mu}_{\mathcal{F}}}_* ( \pi^{\sharp} (\alpha_{\xi})) = \pi_{G^*}^\sharp (\theta_{\xi}),
$$
i.e. it is a $G$-equivariant map.

\end{itemize}
\end{proof}
The first of the equations in the Definition (\ref{prop:gerst}), the equation
\begin{equation}\label{eq:MC}
	d\alpha_{\xi} + \frac{1}{2}\alpha \wedge \alpha \circ \delta(\xi) = 0 ,
\end{equation}
can be solved explicitly in  the case when $M$ is a K\"{a}hler manifold. Before stating the result,  we need to introduce the concept of gauge equivalence of solutions of (\ref{eq:MC}):
\begin{definition}
Two solutions $\alpha$ and $\alpha'$ of eq. (\ref{eq:MC}) are said to be \textbf{gauge equivalent}, if there exists a smooth function $H:M\rightarrow \mk{g}^*$ such that
\begin{equation}
	\alpha' = \exp(ad H)(\alpha)+\int_0^1 dt \exp t(ad H)(dH)
\end{equation}
\end{definition}
\begin{theorem}\label{thm:formal}
Suppose that $M$ is is a K\"{a}hler manifold. The set of gauge equivalence classes of $\alpha\in \Omega^1 (M,\mk{g}^*)$ satisfying the equation
\begin{equation}
	d\alpha_{\xi} +\frac{1}{2}\alpha\wedge\alpha\circ\delta(\xi)=0	
\end{equation}
is in bijective correspondence with the set of the cohomology classes $c\in H^1(M,\mk{g}^*)$ satisfying
\begin{equation}
	[c,c]=0.
\end{equation}
\end{theorem}
\begin{proof}
Since $M$ is is a K\"{a}hler manifold,  $(\Omega^{\bullet}(M),d)$ is a formal CDGA (commutative differential graded algebra) \cite{GDMS}. As a consequence,
\begin{equation}
	Hom({\mathfrak g^*},\Omega^{\bullet} (M)),d,[\cdot ,\cdot ])
\end{equation}
is a formal DGLA and, in particular, there exists a bijection between the equivalence classes of Maurer Cartan elements of $Hom({\mathfrak g^*},\Omega^{\bullet} (M), d ,[\cdot ,\cdot ])$
and Maurer Cartan elements of $Hom({\mathfrak g^*},H_{dR}^{\bullet} (M),[\cdot ,\cdot ])$.

A Maurer-Cartan element in $Hom({\mathfrak g^*},H_{dR}^{\bullet} (M),[\cdot ,\cdot ])$ is an element $c$ in $H^1(M,\mk{g}^*)$ satisfying
\begin{equation}
	[c,c]=0,
\end{equation}
and the claim is proved.
\end{proof}

\subsection{The reconstruction problem}
\label{sec: rec}
In this section we discuss the conditions under which the distribution $\mathcal D$ defined in Theorem \ref{thm: rec}
admits a leaf satisfying eq. (\ref{eq:vanishing}).
In particular, we analyze the case where the structure on $G\st$ is trivial and the Heisenberg group case.
In the following we keep the assumption that $M$ is connected and simply connected.

\subsubsection{The abelian case}
Suppose that $G^*=\mathfrak{g}^*$ is abelian. Then, the forms $\alpha_{\xi}$ satisfy $d\alpha_{\xi}=0$, hence  $\alpha_{\xi}=dH_{\xi}$ (since $H^1 (M)=0$), for some $H_{\xi}\in C^{\infty}(M)$.

Let us denote by $ev_{\xi}$ the linear functions $\mathfrak{g}^* \ni z\rightarrow z(\xi)$. Then $\theta_{\xi}=d(ev_{\xi})$ and the leaves of the distribution $\mathcal D$ coincide with the level sets (on $M\times \mathfrak{g}^*$) of the functions
\begin{equation}
	\{ H_{\xi} - ev_{\xi} \mid \xi\in\mathfrak{g} \}.	
\end{equation}
Furthermore, we have
\begin{equation}
	\varphi (\xi,\eta)(m,z) = \{ H_{\xi} ,H_{\eta} \} - z([\xi,\eta]).	
\end{equation}
In this case, the basic identity (\ref{eq: bff}) reduces to
\begin{equation}
	d\{ H_{\xi} ,H_{\eta} \}=dH_{[\xi,\eta]},
\end{equation}
hence
\begin{equation}
	\{ H_{\xi} ,H_{\eta} \}-H_{[\xi,\eta]}=c(\xi,\eta),
\end{equation}
for some constants $c(\xi,\eta)$. By the Jacobi identity, the constants $c(\xi,\eta)$ define a class $[c]\in H^2 (\mathfrak{g} ,\R)$. Suppose that this class vanishes (for example if $\mathfrak{g}$ semisimple).
Then, there exists a $z_0\in\mathfrak{g}^*$ such that $c(\xi,\eta) = z_0([\xi,\eta])$. Hence, given a leaf $\mathcal{F}$,
\begin{equation}
	\varphi (\xi,\eta)|_{\mathcal{F}} = 0	
\end{equation}
if and only if $\mathcal{F}$ is given by
\begin{equation}
	H_{\xi}-ev_{\xi} - z_0(\xi) = 0.	
\end{equation}
In other words, the space of leaves of $\mathcal D$ which give a momentum map coincides with the affine space modeled on
$\{ z\in \mathfrak{g}^* : z|_{[\mathfrak{g} ,\mathfrak{g} ]}=0\}$ (which again vanishes when $\mathfrak{g}$ is semisimple). This proves the following theorem.
\begin{theorem}
Suppose that $G$ is a connected and simply connected Lie group with trivial Poisson structure and $M$ is compact. Then an infinitesimal momentum map is a map $H:\mathfrak{g} \rightarrow C^\infty (M): \xi\mapsto H_{\xi}$ such that
\begin{equation}
	d\{ H_{\xi} ,H_{\eta}\} = dH_{[\xi,\eta]},\ \quad\forall \xi,\eta\in \mathfrak{g} .
\end{equation}
The element $c(\xi,\eta)=\{ H_{\xi} ,H_{\eta}\} - H_{[\xi,\eta]} $ is a two cocycle $c$ on $\mathfrak{g}$ with values in $\R$. The infinitesimal momentum map $H$ is generated by a momentum map
$\boldsymbol{\mu}$ if this cocycle vanishes and, in this case, $\boldsymbol{\mu}$ is unique.
\end{theorem}

\subsubsection{the Heisenberg group case}
Suppose now that $G^*$ is the Heisenberg group. Let $x,y,z$
be a basis for $\mathfrak{g}^*$, where $z$ is central and $[x,y]=z$. Let $\xi,\eta,\zeta$ be the dual basis of $\mathfrak{g}$. The cocycle $\delta$ on $\mathfrak{g}$ is given by
\begin{equation}
	\delta (\xi)=\delta (\eta)=0 \quad\mbox{ and }\quad \delta (\zeta)=\xi\wedge \eta,
\end{equation}
then
\begin{equation}
	\begin{split}
		d\alpha_{\xi} & = d\alpha_{\eta} = 0\\
		d\alpha_{\zeta}& = \alpha_{\xi}\wedge \alpha_{\eta}
	\end{split}
\end{equation}
There are essentially two possibilities for the Lie bialgebra structure on $\mathfrak{g}^*$, which give the following two possibilities for the Lie algebra structure on $\mathfrak{g}$. Either
\begin{equation}
	[\xi,\eta] = 0 ,\quad [\xi,\zeta] = \xi, \quad [\eta,\zeta] = \eta
\end{equation}
or
\begin{equation}
	[\xi,\eta] = 0, \quad [\xi,\zeta] = \eta, \quad [\eta,\zeta] = -\xi.
\end{equation}
The result below will turn out to be independent of the choice (the computations will be done using the second choice, which corresponds to $G = \R \ltimes \R^2$, with $\R$ acting by rotation on $\R^2$).
Below we use the notation
\begin{equation}
	\delta(\xi) = \sum_i \xi^1_i\wedge \xi^2_i.
\end{equation}
Applying the Cartan formula $\mathcal{L} =[\iota ,d]$ and the identity $[ \alpha_{\xi}, \alpha_{\eta}]_\pi =\alpha_{[\xi,\eta]}$ to the basic equation (\ref{eq: bff}), we get
\begin{equation}
	\sum_i \pi (\alpha_\eta ,\alpha_{\xi_i^1})\alpha_{\xi^2_i} - \sum_i \pi (\alpha_\xi ,\alpha_{\eta^1_i})\alpha_{\eta^2_i} = \alpha_{[\eta,\xi]} - d \pi (\alpha_\eta,\alpha_\xi).
\end{equation}
In our case it gives the following equations
\begin{equation}
	\begin{split}
		d\pi (\alpha_{\xi}, \alpha_{\eta}) & = \alpha_{[\xi,\eta]}\\
		d\pi(\alpha_{\zeta}, \alpha_{\eta}) & = \alpha_{[\zeta,\eta]}+\pi(\alpha_{\eta}, \alpha_{\xi}) \alpha_{\eta}\\
		d\pi (\alpha_{\zeta}, \alpha_{\xi}) & = \alpha_{[\zeta,\xi]}-\pi(\alpha_{\xi} ,\alpha_{\eta})\alpha_{\xi}
	\end{split}
\end{equation}
which are also satisfied after replacing $\alpha$ with $\theta$. Let $\mathcal I $ denote the ideal generating our distribution $\mathcal D$. Then, from above,
\begin{equation}\label{eq:constant}
	d\varphi (\xi,\eta)\in \mathcal I
\end{equation}
and
\begin{equation}
	\varphi (\xi,\eta)|_{\mathcal{F}} = 0 \quad \Longrightarrow \quad d\varphi (\zeta,\eta)|_{\mathcal{F}} \quad\mbox{ and }\quad d\varphi (\zeta,\xi)|_{\mathcal{F}}\in\mathcal I.
\end{equation}
Here, as before, $\mathcal{F}$ is a leaf of $\mathcal D$.
Using the relation (\ref{eq: theta2}), we get
\begin{equation}
	\begin{split}
		& \mathcal{L}_z^*( \pi_{G^*}(\theta_{\xi},\theta_{\eta})) = \mathcal{L}_x^*( \pi_{G^*}(\theta_{\xi},\theta_{\eta})) = \mathcal{L}_y^*( \pi_{G^*}(\theta_{\xi},\theta_{\eta})) = 0\\
		&\mathcal{L}_z^*( \pi_{G^*}(\theta_{\xi},\theta_{\zeta})) = \mathcal{L}_y^*( \pi_{G^*}(\theta_{\xi},\theta_{\zeta})) = 0\\
		&\mathcal{L}_z^*( \pi_{G^*}(\theta_{\eta},\theta_{\zeta})) = \mathcal{L}_x^*( \pi_{G^*}(\theta_{\eta},\theta_{\zeta})) = 0\\
		&\mathcal{L}_x^*( \pi_{G^*}(\theta_{\xi},\theta_{\zeta})) = 1\\
 		&\mathcal{L}_y^*( \pi_{G^*}(\theta_{\eta},\theta_{\zeta})) = 1
	\end{split}
\end{equation}
In particular, $\pi_{G^*} (\theta_{\xi},\theta_{\eta})$ is invariant under left translations. Since $\pi_{G^*}$ is zero at the identity, we get
\begin{equation}\label{eq:piz}
	\pi_{G^*} (\theta_{\xi},\theta_{\eta}) = 0
\end{equation}
Since $ d\varphi (\xi,\eta)\in \mathcal I$, the function $\varphi (\xi,\eta)$ is leafwise constant. Using the definition (\ref{eq:vp}) and the equation (\ref{eq:piz}) it follows that also $\pi(\alpha_{\xi}, \alpha_{\eta})$ is leafwise constant. Hence we have
\begin{lemma}
 $\pi (\alpha_{\xi}, \alpha_{\eta})=c$ is constant on $M$ and necessary condition for existence of the momentum map is $c=0$.
\end{lemma}
By assuming that $c=0$, given a leaf $\mathcal{F}$ by eq. (\ref{eq:constant}),
\begin{equation}
	\varphi(\eta,\zeta)|_{\mathcal{F}} = c_1 \quad \text{and} \quad  \varphi(\xi,\zeta)|_{\mathcal{F}}=c_2
\end{equation}
for some constants $c_1$ and $c_2$. Setting $\mathcal{F}_1 =id \times \exp(c_1 x) \exp (c_2 y)$ to $\mathcal{F}$, we get
\begin{equation}
  \varphi(\eta,\zeta)|_{\mathcal{F}_1} = \varphi(\xi,\zeta)|_{\mathcal{F}_1} = \varphi(\xi,\eta)|_{\mathcal{F}_1} = 0.
\end{equation}
The final result is as follows.
\begin{theorem}\label{thm:heisenberg}
Let $G$ be a Poisson Lie group acting on a Poisson manifold $M$ with an infinitesimal momentum map $\alpha$ and such that $G^*$ is the Heisenberg group.
Let $\xi,\eta,\zeta$ denote the basis of $\mathfrak{g}$ dual to the standard basis $x,y,z$ of $\mathfrak{g}^*$, with $z$ central and $[x,y] = z$. Then
\begin{equation}
	\pi(\alpha_{\xi}, \alpha_{\eta}) = c
\end{equation}
where $c$ is a constant on $M$. The form $\alpha$ lifts to a momentum map $\boldsymbol{\mu}: M\to G^*$ if and only if $c=0$. When $c=0$ the set of momentum maps with given $\alpha$ is one dimensional with free
transitive action of $\R$.
\end{theorem}

\section{Infinitesimal deformations of a momentum map}
\label{sec: inf}

In the following we study infinitesimal deformations of a given momentum map.

Let $(M,\pi)$ be a Poisson manifold with a Poisson action of a Poisson Lie group $(G,\pi_G)$ generated by the momentum map $\boldsymbol{\mu}:M \to G^*$. %In order to make the notation cleaner, the infinitesimal generator of this action is here denoted by $L_{\xi}$ (in the previous sections it was denoted by $L_{\xi}$).
In the following we denote by $\exp :{\mathfrak g}^*\to G^*$ the exponential map.
Suppose that $[-\epsilon ,\epsilon ]\ni t\to \boldsymbol{\mu}_t : M \to G^*$, $\epsilon >0$, is a differentiable
path of momentum maps for this action. We can assume that $\boldsymbol{\mu}_t(m)$ is of the form
\begin{equation}\label{eq:mut}
	\boldsymbol{\mu} (m)\exp (t H_m +t^2\lambda(t,m))
\end{equation}
for some differentiable maps $H : M \to {\mathfrak g}^*: m \mapsto H(m) $ and $\lambda:]-\epsilon ,\epsilon [\times M\to G^*$.

\begin{theorem}\label{thm: idef}
In the notation above the following identities hold.

\medskip

\noindent For all $\xi, \eta \in \mathfrak g$,
\begin{align}
	\label{eq: inf1}   X_{\xi} H (\eta)-X_{\eta} H (\xi) & = H( [\xi,\eta])\\
	\label{eq: inf2}   \{ H(\xi),\ \cdot \} & = -X_{ad^*_H\xi}.
\end{align}
\end{theorem}
\begin{proof}
Let us compute
$$
	\beta_\xi = \left.\frac{d}{dt}\right|_{t=0}\langle d\boldsymbol{\mu}_t,\theta_\xi \rangle = \left.\frac{d}{dt}\right|_{t=0}\langle d(\boldsymbol{\mu}\exp (t H)),\theta_\xi \rangle.
$$
First note that
\begin{equation}
	d(\boldsymbol{\mu}\exp (t H)) = (r_{\exp (t H)})_*d\boldsymbol{\mu} + (l_{\boldsymbol{\mu}})_*d \exp (t H),
\end{equation}
where $r$ and $l$ are the right and left multiplication, respectively. Calculating the derivative $\left.\frac{d}{dt}\right|_{t=0}$ we get:
\begin{equation}
\begin{split}
 	\left.\frac{d}{dt}\right|_{t=0}\langle (r_{\exp (t H)})_*d\boldsymbol{\mu},\theta_{\xi} \rangle &=\left.\frac{d}{dt}\right|_{t=0}\langle d\boldsymbol{\mu},(r_{\exp (t H)})^*\theta_{\xi}\rangle\\
	& = \langle d\boldsymbol{\mu},\mc{L}_H\theta_{\xi}\rangle\\
	& = \langle d\boldsymbol{\mu},\theta_{ad^*_H\xi}\rangle = \alpha_{ad^*_H\xi}
\end{split}
\end{equation}
and
\begin{equation}
\begin{split}
 	\left.\frac{d}{dt}\right|_{t=0}\langle (l_{\boldsymbol{\mu}})_*d \exp (t H),\theta_{\xi} \rangle & =\left.\frac{d}{dt}\right|_{t=0}\langle d \exp (t H),(l_{\boldsymbol{\mu}})^*\theta_{\xi} \rangle\\
	& = \left.\frac{d}{dt}\right|_{t=0}\langle d \exp (t H),\theta_{\xi} \rangle
\end{split}
\end{equation}
The differential of the exponential map $\exp: {\mathfrak g}^*\rightarrow G^*$ is a map from the cotangent bundle of ${\mathfrak g}^*$ to the cotangent bundle of $G\st$. It can be trivialized as
$d\exp:{\mathfrak g}^*\times {\mathfrak g}^*\rightarrow G\st\times {\mathfrak g}^* $. Furthermore, $(\exp^{-1},id): G\st\times {\mathfrak g}^*\rightarrow {\mathfrak g}^*\times {\mathfrak g}^*$ hence the map
${\mathfrak g}^*\times {\mathfrak g}^*\rightarrow {\mathfrak g}^*\times {\mathfrak g}^*$ is given by $tH+o(t^2)$. We get
\begin{equation}
 	\left.\frac{d}{dt}\right|_{t=0}\langle d \exp (t H),\theta_{\xi} \rangle = \left.\frac{d}{dt}\right|_{t=0}\langle d(tH+o(t)),\theta_{\xi} \rangle = d\langle H,\theta_{\xi} \rangle = d\langle H,\xi \rangle
\end{equation}
and finally
\begin{equation}\label{eq: beta}
	\beta_\xi = \alpha_{ad^*_H\xi } + dH (\xi ).
\end{equation}
Since $\pi^{\sharp} (\alpha^t_\xi) = X_{\xi}$ is independent of $t$, $\pi^\sharp \beta_\xi =0$ and we get the identity (\ref{eq: inf2}).

In order to prove the relation (\ref{eq: inf1}), recall that, since $\boldsymbol{\mu}_t$ is a family of Poisson maps, one has
\begin{equation}
	\pi (\alpha^t_\xi ,\alpha^t_\eta)(m) = \pi_{G^*}(\theta_\xi ,\theta_\eta)(\boldsymbol{\mu}_t (m)).
\end{equation}
Applying $\left.\frac{d}{dt}\right|_{t=0}$ to both sides, we get
\begin{equation}
	\pi (\beta_\xi ,\alpha_\eta)(m) + \pi (\alpha_\xi ,\beta_\eta)(m) = \mc{L}_H (\pi_{G^*}(\theta_\xi ,\theta_\eta))(\boldsymbol{\mu} (m)).
\end{equation}
Substituting the expression of $\beta$'s (\ref{eq: beta}) and using the following identity
\begin{equation}
 	\mc{L}_H (\pi_{G^*}(\theta_\xi ,\theta_\eta))(\boldsymbol{\mu} (m)) = H([\xi ,\eta ]) + \pi_{G^*}(\theta_{ad^*_H\xi} ,\theta_\eta) + \pi_{G^*}(\theta_\xi ,\theta_{ad^*_H\eta})
\end{equation}
the claimed equality follows.
\end{proof}

\begin{corollary}\label{cor:unique} Suppose that  $M$ is a Poisson manifold with a Poisson action of a compact semisimple Poisson Lie group $G$. Then any infinitesimal deformation of a momentum map $\boldsymbol{\mu} : M \to G^*$ as above is generated by a one parameter family of gauge transformations.
\end{corollary}

\begin{proof}
Since the relation (\ref{eq: inf1}) implies that $H \in H^1({\mathfrak g},C^\infty (M, {\mathfrak g}^*))$ and $G$ is compact semisimple, $H$ is a Lie coboundary, i. e.
there exists a function
$$
	\Phi : M \to {\mathfrak g}^*
$$
such that
\begin{equation}\label{eq:last1}
	X_{\xi}\Phi = H(\xi).
\end{equation}
In particular, it is easy to check that  $X_{ad^*_H\xi}f = \sum X_{\xi'_i}\Phi\,\xi''_i(f)$, where $\delta(\xi)=\sum\xi'_i\otimes \xi''_i$. Now observe that
\begin{equation}
	\xi\{\Phi,f\} = X_{\xi}\pi(d\Phi,df) = (X_{\xi}\pi)(d\Phi,df)+\{X_{\xi}\Phi,f\} + \{\Phi,X_{\xi}f\}
\end{equation}
hence
\begin{equation}\label{eq:last2}
\begin{split}
	\{H(\xi),f\}& = \{X_{\xi}\Phi,f\}=\xi\{\Phi,f\}-(X_{\xi}\pi)(d\Phi,df)-\{\Phi,X_{\xi}f\}\\
	& = \xi\{\Phi,f\}-\delta(\xi)(\Phi,f)-\{\Phi,X_{\xi}f\}\\
	& = \xi\{\Phi,f\}-X_{\xi'}\Phi\,\xi''(f)-\{\Phi,X_{\xi}f\}.
\end{split}
\end{equation}
Substituting the equations \ref{eq:last1} and \ref{eq:last2} in (\ref{eq: inf2}) we get
\begin{equation}
 	\xi\{\Phi,f\} - \{\Phi,X_{\xi}f\} = 0.
\end{equation}
In other words the Hamiltonian vector field associated to $\Phi $ commutes with the group action and is tangent to the derivative of $\boldsymbol \mu_t$ at $t=0$ as claimed.
\end{proof}

\bibliographystyle{plain}

\end{document}